\documentclass[aps,prl,twocolumn,superscriptaddress,longbibliography,nofootinbib]{revtex4-2}
\usepackage{amsmath,amssymb,amsthm,bm,mathtools,booktabs,array,float,microtype,graphicx}
\usepackage[hidelinks]{hyperref}

\newcommand{\Tr}{\operatorname{Tr}}
\newcommand{\id}{\mathbb{I}}
\newcommand{\GPO}{\mathrm{GPO}}
\newcommand{\TO}{\mathrm{TO}}
\newcommand{\E}{\mathcal{E}}
\newcommand{\Ppin}{\mathcal{P}}
\newcommand{\orb}{\mathrm{orb}}
\newcommand{\aware}{\mathrm{aware}}
\newcommand{\e}{\mathrm{e}}
\newcommand{\ket}[1]{|#1\rangle}
\newcommand{\proj}[1]{|#1\rangle\!\langle #1|}
\newcommand{\barD}{\overline D}
\newcommand{\sandD}{\widetilde D}
\newtheorem{theorem}{Theorem}
\newtheorem{lemma}[theorem]{Lemma}
\newtheorem{corollary}[theorem]{Corollary}
\newtheorem{proposition}[theorem]{Proposition}
\theoremstyle{remark}
\newcommand{\GPC}{\mathrm{GPC}}
\newcommand{\Gtwirl}{\mathcal G}
\newcommand{\eps}{\varepsilon}

\begin{document}

\title{Reliability Is Not Free in Universal Quantum Work Extraction}
\author{Shuai Zeng}
\email{zengshuai@cqupt.edu.cn}
\affiliation{School of Communication and Information Engineering, Chongqing University of Posts and Telecommunications, Chongqing 400065, China}
\date{July 18, 2026}

\begin{abstract}
Universal work extraction shows that input-state knowledge is unnecessary to attain the asymptotic free-energy rate. We ask whether this first-order universality extends to reliability, the exponential decay rate of extraction failure. In the work-battery fidelity formulation, we prove that no phase-independent Gibbs-preserving protocol can retain the state-aware Gibbs-preserving exponent throughout a coherent qubit time-translation orbit: at every positive target rate, its worst pointwise exponent is bounded by the corresponding state-aware thermal-operation value. The proof first establishes an exact finite-blocklength identity between phase-robust Gibbs-preserving extraction and state-aware thermal extraction. A trigonometric Remez inequality then upgrades this minimax identity to a pointwise theorem by ruling out exponential localization of the worst phase. For an explicit coherent-qubit family, known-phase Gibbs-preserving extraction is error free, whereas every phase-independent protocol has a finite worst-pointwise exponent. Thus input-state knowledge can be irrelevant to the first-order work rate yet indispensable for optimal exponential reliability.
\end{abstract}

\maketitle

Universal work extraction establishes a striking first-order equivalence. For independent copies of a nonequilibrium state $\rho$, the asymptotically extractable dimensionless work is governed by $D(\rho\Vert\tau)$ relative to the Gibbs state $\tau$ \cite{Brandao2013,Lostaglio2019}. A single state-independent thermal protocol can attain this rate for every finite-dimensional i.i.d. source \cite{WatanabeUniversal2026}. Related formulations address work extraction from incompletely characterized sources \cite{Safranek2023,WatanabeBlackBox2024}. State information is therefore unnecessary for attaining the optimal amount of work.

Reliability resolves a distinct layer of performance. At a positive target dimensionless work rate $r$, write the failure probability as $\varepsilon_n\asymp\e^{-nB}$. Subexponential failure and a positive exponent $B$ are indistinguishable at first order but require radically different blocklengths. Recent state-aware results show that reliability separates Gibbs-preserving operations (GPO) from thermal operations (TO), while leaving open whether a single state-independent protocol can retain the corresponding exponent pointwise \cite{WatanabeUniversal2026,WatanabeReliability2026}. In the work-battery fidelity formulation, their optimal exponents are
\begin{align}
B_{\GPO}^{\aware}(\rho;r)
 &=\sup_{0<\alpha<1}\frac{1-\alpha}{\alpha}
 \bigl[\barD_\alpha(\rho\Vert\tau)-r\bigr],\nonumber\\
B_{\TO}^{\aware}(\rho;r)
 &=\sup_{0<\alpha<1}\frac{1-\alpha}{\alpha}
 \bigl[\sandD_\alpha(\rho\Vert\tau)-r\bigr],
\label{eq:aware}
\end{align}
where $\barD_\alpha$ and $\sandD_\alpha$ are the Petz and sandwiched R\'enyi relative entropies. For coherence between distinct energy eigenspaces, the first exponent can strictly exceed the second \cite{WatanabeReliability2026}, sharpening the established GPO--TO distinction in coherent thermodynamics \cite{Faist2015,LostaglioCoherence2015}. General GPO implementations can moreover require unbounded coherence \cite{TajimaTakagi2025}. The unresolved question is whether a state-independent GPO can retain that advantage.

We resolve this question on a complete one-parameter coherent orbit. Let $\boldsymbol\Lambda=\{(\Lambda_n,w_n)\}$ be any phase-independent GPO protocol sequence extracting at asymptotic rate $\liminf_n w_n/n\ge r>0$. For a complete qubit time-translation orbit $\mathcal O_\rho=\{\rho_\theta\}_{\theta\in[0,2\pi)}$, it satisfies
\begin{equation}
\inf_{\theta}\, b_{\boldsymbol\Lambda}(\rho_\theta;r)
\le B_{\TO}^{\aware}(\rho;r).
\label{eq:main-nogo}
\end{equation}
The bound applies to the asymptotic reliability profile of each fixed phase, rather than only to a phase selected separately at each blocklength. This distinction is essential because the phase maximizing the error may vary with $n$. Because an orbit-aware protocol may use knowledge of $\mathcal O_\rho$, Eq.~\eqref{eq:main-nogo} applies immediately to fully universal protocols. Figure~\ref{fig:main} displays the finite reliability branches and the zero-error separation window.

\begin{figure}[t]
\centering
\includegraphics[width=\columnwidth]{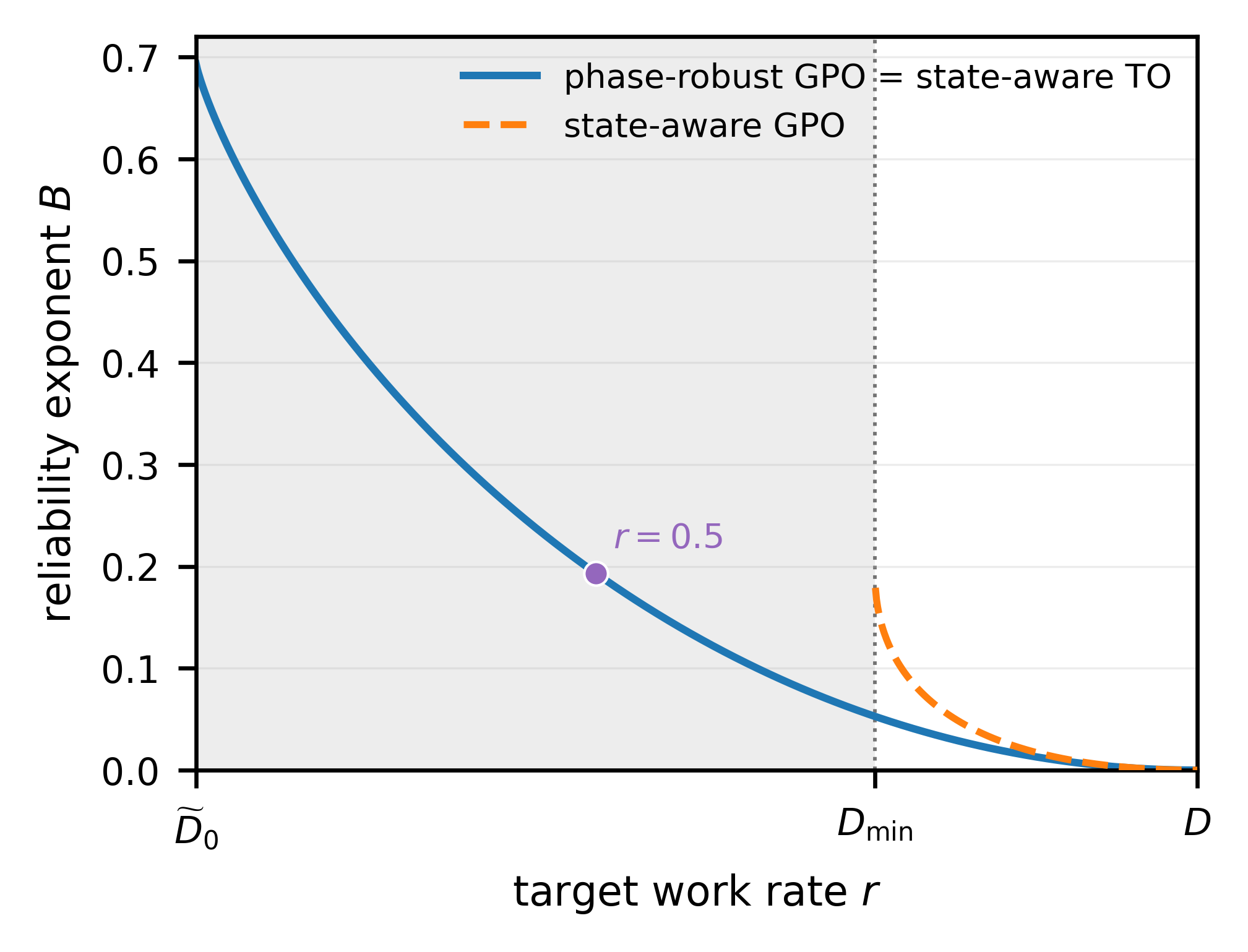}
\caption{Finite reliability exponents for $\tau=\mathrm{diag}(0.8,0.2)$ and $\ket{\psi_\theta}=(\ket0+\e^{i\theta}\ket1)/\sqrt2$. The solid curve is the common phase-robust-GPO and state-aware-TO exponent, and the dashed curve is the state-aware-GPO exponent. In the shaded interval $\widetilde D_0<r<D_{\min}$, state-aware GPO extraction has zero error while the solid branch is finite. For $r<\widetilde D_0$, all three tasks have zero error and lie outside the plotted finite-exponent domain. Both finite branches vanish at the common first-order rate $D$. The purple point marks the worked example $r=0.5$, analyzed below in Eqs.~\eqref{eq:example}--\eqref{eq:infinite-finite}.}
\label{fig:main}
\end{figure}

\emph{Operational setting.---}
Let the nondegenerate qubit Hamiltonian be affinely rescaled so that
\begin{equation}
K=\proj1,\qquad
U_\theta=\e^{-i\theta K}=\proj0+\e^{-i\theta}\proj1,
\label{eq:orbit-unitary}
\end{equation}
and define $\rho_\theta=U_\theta\rho U_\theta^\dagger$. All orbit states have identical spectra, energy populations, free energies, and state-aware reliabilities; only the phase of energetic coherence is unknown. The protocol may use arbitrary fixed auxiliaries and estimate $\theta$ from the input copies, while receiving no side information correlated with the realized phase.

We use the two-level work battery of Ref.~\cite{WatanabeReliability2026}. Extracting dimensionless work $w=\log m$ means converting the input into the excited battery state, whose Gibbs probability is $m^{-1}=\e^{-w}$. A GPO protocol designed to extract $w_n$ at blocklength $n$ induces the success effect
\begin{equation}
\begin{gathered}
M_n=\Lambda_n^\dagger(\proj1),\qquad 0\le M_n\le\id,\\
\Tr(\tau^{\otimes n}M_n)\le\e^{-w_n},
\end{gathered}
\label{eq:effect}
\end{equation}
with phase-dependent error
\begin{equation}
\varepsilon_n(\theta)=1-\Tr(\rho_\theta^{\otimes n}M_n).
\label{eq:error-profile}
\end{equation}
Every effect satisfying Eq.~\eqref{eq:effect} can be saturated to equality and realized by a GPO measure-and-prepare channel. Complete definitions and proofs are given in the Supplemental Material below.

\emph{Exact orbit collapse.---}
At fixed $n$ and $w$, define
\begin{equation}
\E_{\GPO,n}^{\orb}(\mathcal O_\rho;w)
=\inf_{M_n}\sup_\theta\varepsilon_n(\theta).
\label{eq:orbit-error}
\end{equation}
Our first result is the finite-blocklength identity
\begin{equation}
\E_{\GPO,n}^{\orb}(\mathcal O_\rho;w)
=\E_{\TO}(\rho^{\otimes n};w).
\label{eq:collapse}
\end{equation}
for every qubit state, every $n$, and every admissible $w$.

To prove Eq.~\eqref{eq:collapse}, Haar-average an arbitrary feasible effect,
\begin{equation}
\overline M_n=\int_0^{2\pi}\frac{d\phi}{2\pi}
U_\phi^{\otimes n}M_nU_\phi^{\otimes n\dagger}.
\label{eq:effect-twirl}
\end{equation}
Because $\tau^{\otimes n}$ is invariant, feasibility is preserved. The success probability of $\overline M_n$ is phase independent and equals the orbit average of the original success probability; hence its worst-phase performance is at least as good. For qubits, the $U(1)$ twirl is exactly total-energy pinching,
\begin{equation}
\int\frac{d\phi}{2\pi}
U_\phi^{\otimes n}XU_\phi^{\otimes n\dagger}
=\Ppin_{\tau^{\otimes n}}(X).
\label{eq:pinching}
\end{equation}
The orbit problem is therefore exactly GPO extraction from the pinched state, a symmetry-restricted hypothesis-testing problem \cite{HiaiMosonyiHayashi2009}. The one-shot identity of Ref.~\cite{WatanabeReliability2026} equates this quantity with TO extraction from the original state. Taking error exponents gives
\begin{equation}
B_{\GPO}^{\orb}(\mathcal O_\rho;r)
=B_{\TO}^{\aware}(\rho;r).
\label{eq:robust-collapse}
\end{equation}

The pointwise and minimax criteria differ in the order of phase optimization and the asymptotic limit. A fixed sequence is evaluated at each fixed phase through
\begin{equation}
b_{\boldsymbol\Lambda}(\rho_\theta;r)
=\liminf_{n\to\infty}-\frac1n\log\varepsilon_n(\theta),
\label{eq:pointwise}
\end{equation}
whereas the robust exponent places $\sup_\theta$ inside the limit. Consequently, blocklength-dependent localization of the worst-case phase must be controlled before the minimax result can imply a pointwise statement.

\emph{Remez uniformization.---}
Expanding Eq.~\eqref{eq:error-profile} in the excitation basis shows that
\begin{equation}
\varepsilon_n(\theta)=\sum_{k=-n}^{n}c_{n,k}\e^{ik\theta}
\label{eq:fourier}
\end{equation}
is a real nonnegative trigonometric polynomial of degree at most $n$. Suppose every fixed phase had exponent strictly larger than some $b$, and define
\begin{equation}
G_n=\{\theta:\varepsilon_n(\theta)\le\e^{-nb}\}.
\label{eq:good-set}
\end{equation}
Dominated convergence gives $|G_n|\to2\pi$. Writing $s_n=2\pi-|G_n|\to0$, the trigonometric Remez bound \cite{Erdelyi2018} gives
\begin{equation}
\sup_\theta\varepsilon_n(\theta)
\le\e^{-nb}T_{2n}\!\left(\sec\frac{s_n}{2}\right),
\label{eq:remez}
\end{equation}
and
\begin{equation}
\frac1n\log T_{2n}\!\left(\sec\frac{s_n}{2}\right)
\le2\operatorname{arcosh}\!\left(\sec\frac{s_n}{2}\right)\longrightarrow0.
\label{eq:remez-penalty}
\end{equation}
Thus a pointwise advantage shared by all fixed phases would force the same robust advantage. If the left-hand side of Eq.~\eqref{eq:main-nogo} exceeded $B_{\TO}^{\aware}(\rho;r)$, choosing $b$ strictly between them and using Eqs.~\eqref{eq:remez}--\eqref{eq:remez-penalty} would contradict Eq.~\eqref{eq:robust-collapse}. In particular, whenever
\begin{equation}
B_{\GPO}^{\aware}(\rho;r)>B_{\TO}^{\aware}(\rho;r),
\label{eq:strict-gap}
\end{equation}
no phase-independent GPO protocol can attain the state-aware GPO reliability at every fixed phase.

\emph{Explicit pure-qubit separation.---}
Take
\begin{equation}
\tau=\operatorname{diag}(0.8,0.2),\qquad
\ket{\psi_\theta}=\frac{\ket0+\e^{i\theta}\ket1}{\sqrt2},\qquad r=0.5.
\label{eq:example}
\end{equation}
For known $\theta$, the projector onto $\ket{\psi_\theta}^{\otimes n}$ succeeds with probability one and has Gibbs cost
\begin{equation}
\langle\psi|\tau|\psi\rangle^n=\e^{-nD_{\min}},
\qquad D_{\min}=\log2>r.
\label{eq:gpo-threshold}
\end{equation}
It can be completed to a valid GPO effect while preserving unit success, so
\begin{equation}
B_{\GPO}^{\aware}(\psi;0.5)=+\infty.
\label{eq:gpo-infinite}
\end{equation}
For TO, total-energy pinching reduces the finite-blocklength problem to a Neyman--Pearson test over Dicke sectors. Its exact fractional-knapsack solution and large-deviation limit give
\begin{equation}
B_{\TO}^{\aware}(\psi;r)=D(q_r\Vert1/2),\qquad
q_r=\frac{r+\log0.8}{\log4},
\label{eq:closed-example}
\end{equation}
for $-\log0.8<r<D(\psi\Vert\tau)$. At $r=0.5$,
\begin{align}
q_r&=0.1997097128\ldots,\nonumber\\
D(q_r\Vert1/2)&=0.19314744399\ldots.
\label{eq:example-value}
\end{align}
Combining Eqs.~\eqref{eq:robust-collapse} and \eqref{eq:closed-example},
\begin{equation}
\begin{aligned}
B_{\GPO}^{\aware}&=+\infty,\\
B_{\GPO}^{\orb}=B_{\TO}^{\aware}&=0.19314744399\ldots .
\end{aligned}
\label{eq:infinite-finite}
\end{equation}

The separation is not tied to the pure-state zero-error threshold. For every full-rank qubit state with $[\rho,\tau]\ne0$, both state-aware exponents are finite and satisfy
\begin{equation}
B_{\GPO}^{\aware}(\rho;r)>B_{\TO}^{\aware}(\rho;r)>0,
\qquad 0<r<D(\rho\Vert\tau).
\label{eq:mixed-gap}
\end{equation}
Equation~\eqref{eq:main-nogo} therefore imposes a finite but strict universal reliability loss throughout this regime. The same argument extends to finite-dimensional integer-charge $U(1)$ representations whenever the group twirl coincides with thermal pinching: if $\Delta k$ is the single-copy charge width, the $n$-copy error profile has degree at most $n\Delta k$, and the Remez penalty remains subexponential.

The first-order target rate remains unchanged, while exponential precision drops. Set-valued work extraction characterizes minimax exponents under source uncertainty \cite{WatanabeBlackBox2024,ZhangFang2026}, and deterministic coherence no-go theorems address single-shot or constructor-theoretic extraction \cite{PlesnikViolaris2024}. The finite-blocklength identity and Remez uniformization established here reveal a pointwise exponential obstruction: whenever the state-aware GPO--TO gap is strict, every universal protocol loses the known-state GPO exponent on a fixed phase. A universal protocol may reproduce how much work is available without reproducing how reliably it can be extracted. State information is a reliability resource invisible to first-order thermodynamics. More broadly, finite symmetry bandwidth can turn robust boundaries into pointwise impossibility theorems in universal quantum tasks.

\emph{Data availability.---} All supporting numerical values are contained in the Letter and Supplemental Material.

\bibliography{references}

\clearpage
\onecolumngrid
\begingroup
\centering
{\large\bfseries Supplemental Material for\\
``Reliability Is Not Free in Universal Quantum Work Extraction''\par}
\vspace{0.7em}
{Shuai Zeng\par}
{\small School of Communication and Information Engineering, Chongqing University of Posts and Telecommunications, Chongqing 400065, China\par}
{\small zengshuai@cqupt.edu.cn\par}
\endgroup
\vspace{1em}

\setcounter{section}{0}
\setcounter{subsection}{0}
\setcounter{equation}{0}
\setcounter{figure}{0}
\setcounter{table}{0}
\setcounter{theorem}{0}
\renewcommand{\thesection}{S\arabic{section}}
\renewcommand{\thesubsection}{\thesection.\arabic{subsection}}
\renewcommand{\theequation}{S\arabic{equation}}
\renewcommand{\thefigure}{S\arabic{figure}}
\renewcommand{\thetable}{S\arabic{table}}
\renewcommand{\thetheorem}{S\arabic{theorem}}
\renewcommand{\theHequation}{S.\arabic{equation}}
\renewcommand{\theHsection}{S.\arabic{section}}
\renewcommand{\theHsubsection}{S.\arabic{section}.\arabic{subsection}}
\renewcommand{\theHfigure}{S.\arabic{figure}}
\renewcommand{\theHtable}{S.\arabic{table}}
\setcounter{tocdepth}{2}
\tableofcontents
\vspace{0.5em}

\section{Overview}
This Supplemental Material supplies the operational definitions and complete derivations underlying the Letter. It first establishes the hypothesis-testing representation of GPO extraction and defines the state-aware, pointwise, and robust reliability quantities. It then proves the exact finite-blocklength equality between phase-robust GPO extraction and state-aware TO extraction, derives the Fourier bandwidth of the phase-dependent error, and applies a trigonometric Remez inequality to obtain the pointwise no-go theorem. Subsequent sections treat mixed coherent states, the pure-qubit analytic example, finite-blocklength values, and a conditional integer-charge extension.

\section{Operational setting}
\subsection{Thermal state and free operations}
Let $H$ be a finite-dimensional Hamiltonian and $\beta>0$. The Gibbs state is
\begin{equation}
\tau=\frac{\e^{-\beta H}}{Z},\qquad Z=\Tr\e^{-\beta H}.
\end{equation}
A channel is Gibbs preserving if it maps the input Gibbs state to the output Gibbs state. Thermal operations are realized, up to closure, by adjoining a Gibbs ancilla, applying an energy-conserving unitary, and discarding a subsystem. Gibbs-preserving covariant operations impose Gibbs preservation and time-translation covariance directly \cite{Faist2015,LostaglioCoherence2015,TajimaTakagi2025}. Thus
\begin{equation}
\TO\subseteq\GPC\subseteq\GPO.
\end{equation}
All logarithms are natural. For $n$ noninteracting copies, $\tau_n=\tau^{\otimes n}$.

\subsection{Work battery and error convention}
For $m\ge1$, the two-level work battery has Gibbs state
\begin{equation}
\mu_m=\frac{m-1}{m}\proj0+\frac1m\proj1,
\qquad w=\beta W=\log m.
\end{equation}
The target output is $\proj1$. Since the target is pure, squared fidelity equals its excited-state population. Hence the one-shot error is a failure probability,
\begin{equation}
\E_{\mathbb O}(\rho;w)
=1-\sup_{\Lambda\in\mathbb O}F(\Lambda(\rho),\proj1).
\end{equation}
The reliability statements below use this work-battery and fidelity-error formulation.

\section{GPO extraction as hypothesis testing}
\subsection{Success effects}
A GPO protocol induces the Heisenberg-picture success effect
\begin{equation}
M=\Lambda^\dagger(\proj1),\qquad 0\le M\le\id,
\end{equation}
with
\begin{equation}
\Tr(\tau M)=m^{-1}=\e^{-w},
\qquad p_{\rm succ}(\rho,M)=\Tr(\rho M).
\end{equation}
Conversely, any effect satisfying this equality defines the measure-and-prepare GPO
\begin{equation}
\Lambda_M(\omega)=\Tr(M\omega)\proj1+
\Tr[(\id-M)\omega]\proj0.
\end{equation}
Therefore the GPO optimization is exactly a binary asymmetric hypothesis test with a Gibbs type-II budget.

\subsection{Saturating the Gibbs budget}
Hypothesis testing naturally gives $\Tr(\tau M)\le m^{-1}$. The inequality can always be saturated without lowering success. Define
\begin{equation}
\lambda=\frac{m^{-1}-\Tr(\tau M)}{1-\Tr(\tau M)},
\qquad M'=M+\lambda(\id-M).
\end{equation}
Then $0\le M'\le\id$, $\Tr(\tau M')=m^{-1}$, and $M'\ge M$.

\subsection{Imported one-shot identities}
The hypothesis-testing divergence is
\begin{equation}
D_H^\eps(\rho\Vert\sigma)
=-\log\inf\{\Tr(\sigma M):0\le M\le\id,
\Tr[\rho(\id-M)]\le\eps\}.
\end{equation}
Let $\Ppin_\tau$ denote pinching in the distinct eigenspaces of $\tau$. The one-shot identities used in the proof are \cite{WatanabeReliability2026}
\begin{align}
\E_{\GPO}(\rho;w)
&=\min\{\eps:w\le D_H^\eps(\rho\Vert\tau)\},\\
\E_{\GPC}(\rho;w)=\E_{\TO}(\rho;w)
&=\min\{\eps:w\le D_H^\eps(\Ppin_\tau(\rho)\Vert\tau)\}.
\end{align}

\section{Reliability notions}
\subsection{State-aware reliability}
For $\mathbb O\in\{\GPO,\TO\}$, a fixed state $\rho$, and $r>0$, define
\begin{equation}
B_{\mathbb O}^{\aware}(\rho;r)=
\sup_{\{w_n\}}
\left\{\liminf_{n\to\infty}-\frac1n\log
\E_{\mathbb O}(\rho^{\otimes n};w_n):
\liminf_n\frac{w_n}{n}\ge r\right\}.
\end{equation}
Existing results give \cite{WatanabeReliability2026}
\begin{align}
B_{\GPO}^{\aware}(\rho;r)
&=\sup_{0<\alpha<1}\frac{1-\alpha}{\alpha}
[\barD_\alpha(\rho\Vert\tau)-r],\\
B_{\TO}^{\aware}(\rho;r)
&=\sup_{0<\alpha<1}\frac{1-\alpha}{\alpha}
[\sandD_\alpha(\rho\Vert\tau)-r].
\end{align}

\subsection{Qubit phase orbit}
After an affine rescaling of the nondegenerate qubit Hamiltonian,
\begin{equation}
K=\proj1,\qquad U_\theta=\proj0+\e^{-i\theta}\proj1,
\qquad \rho_\theta=U_\theta\rho U_\theta^\dagger.
\end{equation}
The candidate set is the complete orbit
\begin{equation}
\mathcal O_\rho=\{\rho_\theta:\theta\in[0,2\pi)\}.
\end{equation}
The protocol may use the known orbit, arbitrary fixed auxiliaries, and internal phase estimation, while receiving no side information correlated with the realized phase.

For $n$ copies, Haar twirling gives
\begin{equation}
\Gtwirl_n(X)=\int_0^{2\pi}\frac{d\phi}{2\pi}
U_\phi^{\otimes n}XU_\phi^{\otimes n\dagger}
=\sum_{N=0}^{n}\Pi_NX\Pi_N
=\Ppin_{\tau_n}(X),
\end{equation}
where $\Pi_N$ projects onto the total-$N$ excitation subspace. This is the symmetry-restricted hypothesis-testing structure associated with the $U(1)$ action \cite{HiaiMosonyiHayashi2009}.

\subsection{Pointwise and robust quantities}
Let a phase-independent protocol sequence consist of effects $M_n$ satisfying
\begin{equation}
0\le M_n\le\id,\qquad
\Tr(\tau_nM_n)\le\e^{-w_n},
\qquad \liminf_n\frac{w_n}{n}\ge r>0.
\end{equation}
Its error profile and fixed-phase exponent are
\begin{equation}
\eps_n(\theta)=1-\Tr(\rho_\theta^{\otimes n}M_n),
\qquad
b(\rho_\theta)=\liminf_n-\frac1n\log\eps_n(\theta).
\end{equation}
The finite-blocklength orbit error and robust exponent are
\begin{align}
\E_{\GPO,n}^{\orb}(\mathcal O_\rho;w_n)
&=\inf_{M_n}\sup_\theta\eps_n(\theta),\\
B_{\GPO}^{\orb}(\mathcal O_\rho;r)
&=\sup\left\{\liminf_n-\frac1n\log\sup_\theta\eps_n(\theta):
\liminf_n\frac{w_n}{n}\ge r\right\}.
\end{align}
The pointwise and robust exponents have different orders of phase optimization and $\liminf$, which necessitates the uniformization argument below.

\section{Exact finite-blocklength orbit collapse}
\begin{theorem}[Exact orbit collapse]
For every qubit state $\rho$, every $n$, and every target work $w_n$,
\begin{equation}
\E_{\GPO,n}^{\orb}(\mathcal O_\rho;w_n)
=\E_{\TO}(\rho^{\otimes n};w_n).
\end{equation}
\end{theorem}

\begin{proof}
Let $M_n$ be any feasible GPO success effect and define its Haar average
\begin{equation}
\overline M_n=\int_0^{2\pi}\frac{d\phi}{2\pi}
U_\phi^{\otimes n}M_nU_\phi^{\otimes n\dagger}.
\end{equation}
Because $\tau_n$ is invariant, $\Tr(\tau_n\overline M_n)=\Tr(\tau_nM_n)$, so feasibility is preserved. For every $\theta$,
\begin{align}
\Tr(\rho_\theta^{\otimes n}\overline M_n)
&=\int_0^{2\pi}\frac{d\phi}{2\pi}
\Tr(\rho_{\theta-\phi}^{\otimes n}M_n).
\end{align}
The right-hand side is independent of $\theta$ and equals the orbit-average success probability of $M_n$. Thus
\begin{equation}
\sup_\theta\eps_n^{\overline M}(\theta)
=\int_0^{2\pi}\frac{d\phi}{2\pi}\eps_n^M(\phi)
\le\sup_\theta\eps_n^M(\theta).
\end{equation}
The minimax optimization may therefore be restricted to invariant effects.

For invariant $\overline M_n$,
\begin{equation}
\Tr(\rho_\theta^{\otimes n}\overline M_n)
=\Tr[\Ppin_{\tau_n}(\rho^{\otimes n})\overline M_n].
\end{equation}
Optimizing invariant tests is exactly GPO extraction from the pinched state. The one-shot identities above equate this quantity with TO extraction from the original state, giving one direction.

For the converse direction, choose a TO protocol optimal for one orbit representative. Thermal operations are time-translation covariant, while the battery success projector is an energy eigenprojector. Consequently the same protocol has identical success probability for every $\theta$. Since every TO is a GPO, the TO error is achievable uniformly over the orbit. Combining both directions proves the equality.
\end{proof}

\begin{corollary}[Robust reliability]
For every $r>0$,
\begin{equation}
B_{\GPO}^{\orb}(\mathcal O_\rho;r)
=B_{\TO}^{\aware}(\rho;r).
\end{equation}
\end{corollary}

\section{Fourier bandwidth}
\begin{lemma}[Linear Fourier bandwidth]
For every effect $M_n$,
\begin{equation}
\eps_n(\theta)=\sum_{k=-n}^{n}c_{n,k}\e^{ik\theta}
\end{equation}
is a real nonnegative trigonometric polynomial of degree at most $n$.
\end{lemma}

\begin{proof}
In the computational basis, a matrix element between strings $x$ and $y$ acquires phase $\e^{-i(|x|-|y|)\theta}$. The Hamming-weight difference lies in $[-n,n]$, so no Fourier mode outside that range occurs. Any collective, adaptive, randomized, or ancilla-assisted implementation reduces, for the final binary success event, to one effect $M_n$ and therefore obeys the same bound.
\end{proof}

\section{Remez uniformization}
\begin{theorem}[Trigonometric Remez bound] \cite{Erdelyi2018}
Let $Q$ be a trigonometric polynomial of degree at most $N$. If $|Q|\le1$ on a set of measure at least $2\pi-s$, with $0<s<\pi$, then
\begin{equation}
\max_\theta|Q(\theta)|\le T_{2N}\!\left(\sec\frac{s}{2}\right),
\end{equation}
where $T_k$ is the Chebyshev polynomial. Reference~\cite{Erdelyi2018} proves this bound directly for arbitrary complex trigonometric polynomials; its sharper $T_{2N}(\sec(s/4))$ form applies to even trigonometric polynomials.
\end{theorem}

\begin{lemma}[Pointwise-to-minimax uniformization]
If
\begin{equation}
\inf_\theta\liminf_{n\to\infty}-\frac1n\log\eps_n(\theta)>b,
\end{equation}
then
\begin{equation}
\liminf_{n\to\infty}-\frac1n\log\sup_\theta\eps_n(\theta)\ge b.
\end{equation}
\end{lemma}

\begin{proof}
Define
\begin{equation}
G_n=\{\theta:\eps_n(\theta)\le\e^{-nb}\}.
\end{equation}
The premise implies that every fixed $\theta$ belongs to $G_n$ for all sufficiently large $n$. Dominated convergence therefore gives $|G_n|\to2\pi$. Let $s_n=2\pi-|G_n|\to0$. Applying the Remez bound to $Q_n(\theta)=\e^{nb}\eps_n(\theta)$ yields
\begin{equation}
\sup_\theta\eps_n(\theta)
\le\e^{-nb}T_{2n}\!\left(\sec\frac{s_n}{2}\right).
\end{equation}
Since
\begin{equation}
\frac1n\log T_{2n}\!\left(\sec\frac{s_n}{2}\right)
\le2\operatorname{arcosh}\!\left(\sec\frac{s_n}{2}\right)\to0,
\end{equation}
the amplification is $\e^{o(n)}$, which proves the lemma.
\end{proof}

\section{Pointwise no-go}
\begin{theorem}[Pointwise no-go]
For every phase-independent GPO sequence with target rate $r>0$,
\begin{equation}
\inf_{\theta\in[0,2\pi)}
\liminf_{n\to\infty}-\frac1n\log\eps_n(\theta)
\le B_{\TO}^{\aware}(\rho;r).
\end{equation}
\end{theorem}

\begin{proof}
If the pointwise infimum were strictly above the TO exponent, choose an intermediate $b$. The Fourier-bandwidth and Remez lemma would force the robust exponent of the same protocol to be at least $b$, contradicting the exact robust boundary.
\end{proof}

\begin{corollary}
Whenever
\begin{equation}
B_{\GPO}^{\aware}(\rho;r)>B_{\TO}^{\aware}(\rho;r),
\end{equation}
no phase-independent GPO protocol attains the state-aware GPO exponent at every phase. More precisely, for every
\begin{equation}
0<\delta<B_{\GPO}^{\aware}-B_{\TO}^{\aware},
\end{equation}
there exists a fixed phase $\theta_\delta$ such that
\begin{equation}
b(\rho_{\theta_\delta})<B_{\TO}^{\aware}+\delta
<B_{\GPO}^{\aware}.
\end{equation}
\end{corollary}

\section{Mixed coherent qubits}
\begin{proposition}
If $\rho>0$, $[\rho,\tau]\ne0$, and $0<r<D(\rho\Vert\tau)$, then
\begin{equation}
B_{\GPO}^{\aware}(\rho;r)>B_{\TO}^{\aware}(\rho;r)>0,
\end{equation}
and both exponents are finite.
\end{proposition}

\begin{proof}
For full-rank $\rho$, both zero-order R\'enyi limits vanish. The variational objectives tend to $-\infty$ as $\alpha\downarrow0$, tend to zero as $\alpha\uparrow1$, and are positive at an interior point because $r<D(\rho\Vert\tau)$. Thus both optima occur in the interior. The Araki--Lieb--Thirring inequality is strict for noncommuting $\rho$ and $\tau$, which yields the strict finite GPO--TO gap.
\end{proof}

\section{Pure-qubit analytic solution}
\subsection{General orbit and GPO zero-error threshold}
Let
\begin{equation}
\tau=t_0\proj0+t_1\proj1,\qquad
\ket{\psi_\theta}=\sqrt{1-p}\ket0+\e^{i\theta}\sqrt p\ket1.
\end{equation}
For the pure state,
\begin{equation}
D_{\min}(\psi\Vert\tau)
=-\log[(1-p)t_0+pt_1].
\end{equation}
Whenever $r<D_{\min}$, the state projector has Gibbs cost below $\e^{-nr}$ and yields exact GPO success. Hence
\begin{equation}
B_{\GPO}^{\aware}(\psi;r)=\infty.
\end{equation}

\subsection{Dicke sectors}
The $n$-copy state decomposes as
\begin{equation}
\ket{\psi_\theta}^{\otimes n}
=\sum_{N=0}^n\e^{iN\theta}\sqrt{P_N}\ket{D_N^{(n)}},
\qquad P_N=\binom nN(1-p)^{n-N}p^N.
\end{equation}
After total-energy pinching,
\begin{equation}
\Omega_{p,n}=\sum_{N=0}^nP_N\proj{D_N^{(n)}}.
\end{equation}
The Gibbs weight of the normalized Dicke direction is
\begin{equation}
Q_N=t_0^{n-N}t_1^N.
\end{equation}
There is no binomial factor in $Q_N$ because the test accepts one normalized direction, not the entire energy shell.

\subsection{Finite-blocklength optimization}
The optimal effect has the form
\begin{equation}
M_n=\sum_{N=0}^nx_N\proj{D_N^{(n)}},\qquad0\le x_N\le1,
\end{equation}
and solves
\begin{equation}
\max_x\sum_NP_Nx_N
\quad\text{subject to}\quad
\sum_NQ_Nx_N\le\e^{-nr}.
\end{equation}
This is a fractional-knapsack/Neyman--Pearson problem: accept sectors in decreasing order of $P_N/Q_N$, with at most one fractionally accepted boundary sector.

\subsection{Large-deviation exponent}
With $N/n\to q$,
\begin{equation}
P_N=\exp[-nD(q\Vert p)+O(\log n)],
\qquad Q_N=\e^{-nc(q)},
\end{equation}
where
\begin{equation}
c(q)=-(1-q)\log t_0-q\log t_1.
\end{equation}
Therefore
\begin{equation}
B_{\TO}^{\aware}(\psi;r)
=\inf_{q:c(q)\le r}D(q\Vert p).
\end{equation}
For $c_0<r<c(p)$,
\begin{equation}
q_r=\frac{r+\log t_0}{\log(t_0/t_1)},
\qquad B_{\TO}^{\aware}(\psi;r)=D(q_r\Vert p).
\end{equation}
The complete open-interval structure is
\begin{equation}
B_{\TO}^{\aware}(\psi;r)=
\begin{cases}
\infty,&0<r<-\log t_0,\\
D(q_r\Vert p),&-\log t_0<r<D(\psi\Vert\tau),\\
0,&r\ge D(\psi\Vert\tau).
\end{cases}
\end{equation}

\subsection{Explicit pure-qubit separation}
Take
\begin{equation}
\tau=\operatorname{diag}(0.8,0.2),\qquad
\ket{\psi_\theta}=\frac{\ket0+\e^{i\theta}\ket1}{\sqrt2},
\qquad r=0.5.
\end{equation}
Then
\begin{equation}
D_{\min}=\log2>0.5,
\end{equation}
so the known-phase GPO error is exactly zero. Meanwhile,
\begin{equation}
q_r=0.19970971277855967790\ldots,
\end{equation}
and
\begin{equation}
B_{\TO}^{\aware}=D(q_r\Vert1/2)
=0.19314744398892003976\ldots.
\end{equation}
Thus
\begin{equation}
B_{\GPO}^{\aware}=\infty,
\qquad B_{\GPO}^{\orb}=B_{\TO}^{\aware}
=0.193147443988920\ldots.
\end{equation}

\subsection{Finite-blocklength values}
The exact fractional-knapsack values are listed below.
\begin{table}[H]
\centering
\begin{tabular}{rcc}
\toprule
$n$ & $\eps_n^*$ & $-\log(\eps_n^*)/n$\\
\midrule
5 & $8.281517291279069\times10^{-2}$ & 0.4982287973279442\\
10 & $2.521319888781785\times10^{-2}$ & 0.3680387656192378\\
20 & $2.791259586817086\times10^{-3}$ & 0.2940631160109723\\
50 & $5.662818071729281\times10^{-6}$ & 0.2416317779424973\\
100 & $2.587147740495017\times10^{-10}$ & 0.2207529491955584\\
200 & $7.345346108718376\times10^{-19}$ & 0.2087752491834705\\
500 & $2.749802857904568\times10^{-44}$ & 0.2006044297413279\\
1000 & $1.986950864522377\times10^{-86}$ & 0.1973357167626979\\
$\infty$ & --- & $0.1931474439889200\ldots$\\
\bottomrule
\end{tabular}
\end{table}
The convergence from above contains the usual $O(\log n/n)$ type and boundary-sector corrections.

\section{Conditional integer-charge extension}
Let
\begin{equation}
K=\sum_jk_j\Pi_j,\qquad k_j\in\mathbb Z,
\qquad \Delta k=\max_jk_j-\min_jk_j.
\end{equation}
The $n$-copy error then has Fourier degree at most $n\Delta k$, and
\begin{equation}
\frac1n\log T_{2n\Delta k}\!\left(\sec\frac{s_n}{2}\right)
\le2\Delta k\operatorname{arcosh}\!\left(\sec\frac{s_n}{2}\right)\to0.
\end{equation}
Whenever the group twirl coincides with thermal pinching, the exact orbit collapse and pointwise no-go follow in the same form.

\end{document}